\documentclass[12pt]{article}

\usepackage{amsmath}
\usepackage{amssymb}
\usepackage{nccmath}
\usepackage{colonequals}
\usepackage{enumitem}
\usepackage{empheq}

\usepackage{amsthm}
\usepackage{multicol}
\theoremstyle{plain}
\newtheorem{theorem}{Theorem}
\newtheorem{lemma}[theorem]{Lemma}
\newtheorem{corollary}[theorem]{Corollary}

\theoremstyle{definition}
\newtheorem{definition}[theorem]{Definition}
\newtheorem{example}[theorem]{Example}

\newcommand{\hilbertrule}[2]{\displaystyle{\frac{#1}{#2}}}

\newcommand{\justifies}{{\mathbin{:}}}
\newcommand{\bang}{\mathop{!}}
\newcommand{\eFml}{\mathsf{eFml}}
\newcommand{\Fml}{\mathsf{Fml}}
\newcommand{\Tm}{\mathsf{Tm}}
\newcommand{\Prop}{\mathsf{Prop}}
\newcommand{\Comp}{\mathsf{Comp}}
\newcommand{\I}{\mathsf{I}}
\newcommand{\pow}{\mathcal{P}}
\newcommand{\QU}{\mathbb{Q}\cap [0,1]}
\newcommand{\IPJ}{\mathsf{IPJ}_{\I}}

\begin{document}
\title{A logic of interactive proofs}
\author{David Lehnherr \and Zoran Ognjanovi\'c\thanks{Supported by the Science  Fund of the Republic of Serbia project AI4TrustBC.} \and Thomas Studer\thanks{Supported by the Swiss National Science Foundation grant 200020\_184625.}}

\date{}
\maketitle              

\begin{abstract}
We introduce the probabilistic two-agent justification logic $\mathsf{IPJ}$,  a logic in which we can reason about agents that perform interactive proofs. 
In order to study the growth rate of the probabilities in $\mathsf{IPJ}$,  we present a new method of parametrizing $\mathsf{IPJ}$ over certain negligible functions.
Further, our approach leads to a new notion of zero-knowledge proofs.
\end{abstract}

\section{Introduction}
An interactive proof system~\cite{Babai85,Goldwasser85} is a protocol between two agents, the prover and the verifier. The aim of the protocol is that the prover can  prove  its  knowledge of  a secret  to the verifier.  
To achieve this, the prover must answer  a challenge provided by the verifier. Usually, the protocols are such that the verifier only knows with high probability that the prover knows the secret, that is the probability is a negligible function in the length of the challenge. 

Several formalizations of the notion \emph{proof of knowledge} are compared and analyzed in~\cite{10.1007/3-540-48071-4_28}.  The aim of the present paper is to provide an epistemic logic model for interactive proofs of knowledge.

Our logic of interactive proofs and justifications $\IPJ$ will be a combination of modal logic, justification logic, and probabilistic logic. The logic includes two agents, $P$ (the prover) and $V$ (the verifier).
The modal part of $\IPJ$ consists of two  S4 modalities $\Box_P$ and $\Box_V$. As usual, $\Box_a$ means \emph{agent $a$ knows that}.
Justification logic adds explicit reasons for the agents' knowledge~\cite{artemovFittingBook,jlbook}. We have formulas of the form $t \justifies_a \alpha$, which stand for \emph{agent $a$ knows $\alpha$ for reason $t$}. The reason represented by the term $t$, can be a formal proof as in the first justification logic, the Logic of Proofs~\cite{Art01BSL,weakArithm}, the execution of an interactive proof protocol, the result of an agent's reasoning, or any other justification of knowledge like, e.g., direct observation. For $\IPJ$, we will use a two-agent version of the logic of proofs together with the justification yields belief principle $t \justifies_a \alpha \to \Box_a \alpha$.
The third ingredient of $\IPJ$ are probability operators of the form $\mathcal{P}_{\geq r}$ and  $\mathcal{P}_{\approx r}$ meaning \emph{with probability  greater than or equal to $r$} and \emph{with probability  approximately~$r$}, respectively. 
For the probabilistic part, we use the approach of~\cite{ogra00,ognjanovicRM16}, which has been adapted to justification logic in~\cite{komaogst,koogstJournal}. In order to deal with approximate probabilities, we need probability measures that can take non-standard values. Logics of this kind have been investigated in~\cite{OSS17,approxCondDefault}.

Goldwasser et al.~\cite{Goldwasser85} introduced interactive proof systems as follows.
Let $\mathcal{L}$ be a language and
$P$ and $V$ a pair of interacting (probabilistic) Turing machines, where $P$ has unrestricted computational  power and $V$ is polynomial time.
$\langle P,V\rangle$ is an interactive proof system for $\mathcal{L}$ if the following conditions hold:
\begin{enumerate}
\item \textsf{Completeness: }  For all $k \in \mathbb{N}$, there exists an $m \in \mathbb{N}$ such that for all inputs $x \in \mathcal{L}$ with $|x|>m$, the probability of $\langle P,V\rangle$ accepting $x$ is at least $1-|x|^{-k}$. 
\item \textsf{Soundness:} For all $k \in \mathbb{N}$, there exists an $m \in \mathbb{N}$ such that for all inputs $x  \not\in \mathcal{L}$ with $|x|>m$ and any interactive Turing machine $P'$, the probability of $\langle P',V\rangle$ accepting $x$ is at most $|x|^{-k}$.  
\end{enumerate}

Less formally,  the agent $P$ tries to prove its knowledge about a proposition~$\alpha$ to the agent $V$. They may do that by following a challenge-response scheme. That is, $V$ sends a challenge to $P$ who then tries to answer it using his knowledge about $\alpha$. 
On success, $V$'s confidence in $P$ knowing $\alpha$ is increased. Moreover, the harder the challenge, the stronger is $V$'s belief. However, $P$ may be dishonest and hence $V$ may be convinced (with a low probability) that a wrong statement is true.

In order to model this in $\IPJ$, we introduce terms of the form $f^n_t$ that represents $V$'s view of the run of the protocol where $P$ has evidence $t$ and $n$ is a measure for the complexity of the run (this may refer to the complexity of the challenge in a challenge response scheme). The outcome of a run will be formalized as $\mathcal{P}_{\geq r}(f^n_t\justifies_V\Box_P\alpha)$ meaning that with probability greater than or equal to~$r$, the run of the protocol with complexity $n$ provides a justification for $V$ that~$P$ knows $\alpha$.
Note that we are abstracting away the concrete protocol.  Moreover, the subscript $t$ in $f^n_t$ does not imply that $V$ has access to $t$; it only states that $P$'s role in the protocol  depends on $t$.
We say that a formula $\alpha$ is interactively provable if the following two conditions hold:
\begin{enumerate}
\item \textsf{Completeness: } Assume $t\justifies_P \alpha$. For all $k \in \mathbb{N}$, there exists a degree of complexity $m \in \mathbb{N}$ such that, for $n>m$ the probability of $f^n_t$ justifying $\Box_P\alpha$ from $V$'s view is at least $1-n^{-k}$. 
\item \textsf{Soundness:} Assume $\neg t\justifies_P \alpha$. For all $k \in \mathbb{N}$, there exists a degree of complexity $m \in \mathbb{N}$ such that, for $n>m$ the probability of $f^n_t$ justifying $\Box_P\alpha$ from $V$'s view is at most $n^{-k}$. 
\end{enumerate}
Since $\IPJ$ is a propositional logic, we need a way to express the soundness and completeness condition without quantifiers.
For integers $m,k$, we start with sets of formulas $\I_{m,k}$ and define  the set of interactively provable formulas 
\[
\I := \bigcap_{k}\bigcup_{m}\I_{m,k}.
\] 
If a formula $\alpha$ belongs to $\I_{m,k} $, then the following two conditions must hold for $n>m$:
\begin{enumerate}
\item $ t\justifies_P\alpha \rightarrow \mathcal{P}_{\geq 1-\frac{1}{n^k}}(f^n_t\justifies_V \Box_P \alpha)$
\item $\neg(t\justifies_P\alpha) \rightarrow \mathcal{P}_{\leq \frac{1}{n^k}}(f^n_t\justifies_V \Box_P \alpha)$
\end{enumerate}
Therefore, if $\alpha \in \I$ and $t\justifies_P\alpha$ then, for every $k$, there exists an $m$ such that $\alpha\in \textsf{I}_{m,k}$ and thus $\mathcal{P}_{\geq 1-\frac{1}{n^k}}(f^n_t\justifies_V \Box_P \alpha)$.
Observe that this closely resembles the previously stated completeness property of interactive proof systems. The soundness property is obtained analogously.

Furthermore, we allow the probability operators to take non-standard values and consider protocols with transfinite complexity~$\omega$ to capture the notion of a limit. Hence we can express statements of the form
\[
\text{if $t\justifies_P\alpha$, then the probability of $f^\omega_t\justifies_V\Box_P\alpha$ is almost~$1$.}
\]
Using the operator  $\mathcal{P}_{\approx r}$, we add two more conditions for interactively provable formulas:
\begin{enumerate}[resume]
\item $t\justifies_P\alpha \rightarrow \mathcal{P}_{\approx 1}(f^\omega_t\justifies_V \Box_P \alpha) \text{ if } \alpha \in \textsf{I}$;
\item $\neg(t\justifies_P\alpha) \rightarrow \mathcal{P}_{\approx 0}(f^\omega_t\justifies_V \Box_P \alpha) \text{ if } \alpha \in \textsf{I}$.
\end{enumerate}

We also include a principle saying that the justifications $f^n_t$ are monotone in the complexity $n$:
\begin{enumerate}[resume]
\item $ f^m_t\justifies_a \alpha  \rightarrow f^n_t\justifies_a \alpha$ if $m<n$.
\end{enumerate}

Justification logics with interacting agents are not new. 
Yavorskaya~\cite{TYav08TOCSnonote} introduced the evidence verification operator $\bang^V_P$ that can be used by $V$ to verify $P$'s evidence, i.e.~her system includes the axiom 
$
t\justifies_P \alpha \to \bang^V_P t\justifies_V t\justifies_P \alpha
$.
This resembles the definition of the complexity class NP as interactive proof system, see, e.g.,~\cite{CompComp}. There, the verifier is a deterministic Turing machine. The prover generates a proof certificate $t$ for $\alpha$ (where the complexity of $t$ is polynomial in~$\alpha$), i.e.~we have $t\justifies_P \alpha$. 
Now $P$ sends this certificate $t$ to $V$ and $V$ checks it (which can be done in polynomial time). A successful check results in $\bang^V_P t$ being a justification for $V$ that $P$ knows the proof certificate~$t$ for $\alpha$, i.e.~$\bang^V_P t\justifies_V t\justifies_P \alpha$.

\section{Syntax}

Let $\mathbb{N}$ be the  set of natural numbers and  $\mathbb{N}^+\colonequals\mathbb{N} \setminus \{0\}$. We define
\[
\Comp \colonequals \mathbb{N} \cup \{\omega\}
\]
where $\omega > n$ for each $n \in \mathbb{N}$.

We start with a countable set of justification variables and justification constants. Further we have a symbol $f^n$ for each $n \in \Comp$.
The set of \emph{terms}~$\Tm$ is given by the following grammar
\begin{equation*}
t \coloncolonequals c \mid x \mid t \cdot t\mid t+t \mid \bang t \mid f^n t
\end{equation*}
where $c$ is a justification constant and $x$ is a justification variable. In the following,  we usually write $f^n_t$ for $f^n t$.

Our language is based on two agents, the prover~$P$ and the verifier~$V$. We write $a$ for an arbitrary agent, i.e.~either $P$ or $V$.  Further,  we use a countable set of atomic propositions~$\Prop$.
The set of \emph{epistemic formulas} $\eFml$ is given by the following grammar:
\begin{equation*}
\alpha \coloncolonequals p \mid \lnot\alpha \mid \alpha \land \alpha \mid \Box_a \alpha \mid t \justifies_a \alpha
\end{equation*}
where $p$ is an atomic proposition, $t$ is a term and $a$ is an agent.

For our  formal approach, we  consider probabilities that range over the unit interval of a non-archimedean recursive field that contains all rational numbers. We proceed as in~\cite{approxCondDefault} by choosing the unit interval of the Hardy field $\mathbb{Q}[\epsilon]$. The set $\mathbb{Q}[\epsilon]$ consists of all rational functions of a fixed non-zero infinitesimal $\epsilon \in \mathbb{R}^*$, where $\mathbb{R}^*$ is a non-standard extension of $\mathbb{R}$ (see~\cite{Robinson96}) for further details). Its positive elements have the form:
\begin{equation*}
\epsilon^k\frac{\sum^n_{i=0} a_i\epsilon^i}{\sum^m_{i=0}b_i\epsilon^i},
\end{equation*}
where $a_i,b_i \in \mathbb{Q}$ for all $i\geq0$ and $a_0\cdot b_0 \neq 0$.
We use $S$ to denote the unit interval of $\mathbb{Q}[\epsilon]$.

The set of \emph{formulas} $\Fml$ is given by the following grammar:
\begin{equation*}
A  \coloncolonequals \alpha \mid \mathcal{P}_{\geq s}\alpha \mid \mathcal{P}_{\approx r}\alpha \mid \lnot A \mid A \land A
\end{equation*}
where $\alpha$ is an epistemic formula,  $s \in S$, and $r \in \QU$. 

Since any epistemic formula  is a formula, we sometimes use latin letters to denote epistemic formulas, e.g.~in $t \justifies A \to  \mathcal{P}_{\approx 1}B$, the letters $A$  and $B$ stand  for epistemic formulas.  

The remaining propositional connectives are defined as usual.
Further we use the following syntactical abbreviations:
\begin{align*}
& \mathcal{P}_{<s}\alpha  \text{ denotes } \neg \mathcal{P}_{\geq s}\alpha \quad \quad 
\mathcal{P}_{\leq s}\alpha  \text{ denotes }  \mathcal{P}_{\geq 1- s}\neg \alpha\\
& \mathcal{P}_{>s}\alpha  \text{ denotes } \neg \mathcal{P}_{\leq s}\alpha \quad \quad
 \mathcal{P}_{=s}\alpha  \text{ denotes } \mathcal{P}_{\leq s}\alpha \land \mathcal{P}_{\geq s}\alpha
\end{align*}

Our Logic of Interactive Proofs $\IPJ$ depends on a parameter $\I$. We  will introduce that parameter later when it will be relevant. We start with presenting the axioms of  $\IPJ$, which are divided into three groups: epistemic axioms, probabilistic axioms, interaction axioms.

\subsection*{Epistemic axioms}
For both modal operators $\Box_P$ and $\Box_V$ we have the axioms for the modal logic S4.
\begin{fleqn}
\begin{equation*}
\begin{array}{p{3em}l}
\textsf{(p)} & \text{all propositional tautologies}\\
\textsf{(k)} &  \Box_a (A \to B) \to (\Box_a A \to \Box_a B)\\
\textsf{(t)} &  \Box_a A \to A\\
\textsf{(4)} &  \Box_a A \to \Box_a\Box_a A\\
\end{array}
\end{equation*}
\end{fleqn}
For both agents, we have the axioms for the Logic of Proofs~\cite{Art01BSL} and  the connection axiom~$\textsf{(jyb)}$.
This yields the system $\mathsf{S4LP}$ from~\cite{ArtNog05JLC}.
\begin{fleqn}
\begin{equation*}
\begin{array}{p{3em}l}
\textsf{(j)} &  s \justifies_a ( A \rightarrow B )\rightarrow (t\justifies_a A \rightarrow_a s\cdot t \justifies_a B) \\
\textsf{(j+)} & (s \justifies_a A \lor t \justifies_a A ) \rightarrow (s+t) \justifies_a A\\
\textsf{(jt)} &  t \justifies_a A \rightarrow A \\
\textsf{(j4)} &  t \justifies_a A \rightarrow \bang t \justifies_a t \justifies_a A\\
\textsf{(jyb)} &  t \justifies_a A \rightarrow \Box_a A
\end{array}
\end{equation*}
\end{fleqn}

\subsection*{Probabilistic axioms}
The probabilistic axioms correspond  to the axiomatization of approximate conditional probabilities used in~\cite{OSS17,approxCondDefault} adapted to the unconditional case.
\begin{fleqn}
\begin{equation*}
\begin{array}{p{3em}l}
\textsf{(p1)} & P_{\geq 0} A\\
\textsf{(p2)} & P_{\leq s} A \to P_{< t} A,  \text{ where $s < t$}  \\
\textsf{(p3)} & P_{< s} A \to P_{\leq s} A\\
\textsf{(p4)} & P_{\geq 1} (A \leftrightarrow B) \to  (P_{= s} A \rightarrow P_{= s} B)\\
\textsf{(p5)} & P_{\leq s} A \leftrightarrow P_{\geq 1-s}\neg A \\
\textsf{(p6)} & (P_{=s} A \wedge P_{=t} B \wedge P_{\geq 1}\neg(A \wedge B)) \rightarrow P_{=\min(1,s+t)}(A \vee B)   \\
\textsf{(pa1)} & P_{\approx r} A \to P_{\geq r_1} A , \text{ for every rational $r_1 \in [0,r)$}\\
\textsf{(pa2)} & P_{\approx r} A \to P_{\leq r_1} A ,  \text{ for every rational $r_1 \in (r,1]$} 
\end{array}
\end{equation*}
\end{fleqn}

\subsection*{Interaction axioms}

So far,  we have axioms  for an epistemic justification logic with approximate probabilities. Let us now add axioms for  terms of the form $f^n_t$ that model interactive proof protocols. These axioms depend on the parameter $\I$ in $\IPJ$, which we introduce next.

An \emph{interaction specification} 
 $\I$ is a function $\I:  \mathbb{N} \times \mathbb{N} \to \pow(\eFml)$, i.e.~to each
 $m,k \in \mathbb{N}$ we assign a set of epistemic formulas $\I(m,k)$. 
 In the following,  we  write  $\I_{m,k}$ for $\I(m,k)$. 
 Further, we overload the notation and use $\I$ also to denote the set 
\[
\I \colonequals \bigcap_k\bigcup_m\I_{m,k}.
\]
The interaction axioms are:
\begin{fleqn}
\begin{equation*}
\begin{array}{p{3em}l}
\textsf{(m)} & f^m_t\justifies_a \alpha  \rightarrow f^n_t\justifies_a \alpha  \text{ for all $m,n \in \Comp$ such that $m<n$}\\
\textsf{(c)} & t\justifies_P\alpha  \rightarrow {P}_{\geq 1-\frac{1}{n^k}}(f^n_t\justifies_V \Box_P \alpha) \text{ if $n>m$ and } \alpha \in \I_{m,k}\\
\textsf{(s}) & \neg(t\justifies_P \alpha ) \rightarrow {P}_{\leq \frac{1}{n^k}}(f^n_t\justifies_V \Box_P \alpha) \text{ if $n>m$ and } \alpha \in \I_{m,k}\\
\textsf{(c$\omega$)} & t\justifies_P \alpha  \rightarrow {P}_{\approx 1}(f^\omega_t\justifies_V \Box_P \alpha) \text{ if } \alpha \in \I\\
\textsf{(s$\omega$)} &\neg(t\justifies_P \alpha) \rightarrow {P}_{\approx 0}(f^\omega_t\justifies_V \Box_P \alpha ) \text{ if }  \alpha \in \I
\end{array}
\end{equation*}
\end{fleqn}

\subsection*{Inference rules}

The rules of $\IPJ$ are the following. 
We have modus ponens:
\[
	\hilbertrule{A \quad A\rightarrow B}{B}\
\]
$\IPJ$ also includes the modal necessitation rule as well as the axiom necessitation rule from justification logic:
\[
		\hilbertrule{A}{\Box A} \qquad\qquad
		\hilbertrule{\text{$A$ is an axiom of $\IPJ$}}{c_1 \justifies_{a_1} c_2 \justifies_{a_2} \cdots  c_n \justifies_{a_n} A} 
\]
for arbitrary constants $c_i$ and agents $a_i$.
Of course,  it would be possible to parameterize $\IPJ$ additionally by a constant specification as it is often done in justification logic. 
This would not affect our treatment of interactive proofs.

We have  the following rules for the probabilistic part:
\begin{enumerate}
\item
From $A$ infer  $P_{\geq 1 } A$ 
\item
From $B \to P_{\neq s} A$ for all $s \in S$ infer $B\to \bot$
\item
From $B \to P_{\geq r-\frac{1}{n}} A$ and $B \to P_{\leq r+\frac{1}{n}} A$ for all integers $n$,  infer 
\[
B \to P_{\approx r} A
\]
\end{enumerate}
Of course in the last rule, only premises $B \to P_{\geq r-\frac{1}{n}} A$ are considered for which $r-\frac{1}{n}>0$ holds and $B \to P_{\leq r+\frac{1}{n}} A$ is only considered if $r+\frac{1}{n}<1$.

\section{Semantics}

For this section, we assume that we are given an arbitrary interaction specification $\I$. 
Many notions in this chapter will depend on that parameter.
For any set $X$ we use $\pow(X)$ to denote the power set of~$X$.
We will use a Fitting-style semantics~\cite{Fit05APAL} for justification logic, but modular models~\cite{Art12SLnonote,KuzStu12AiML} would work as well.

\begin{definition}[Evidence relation]\label{def:evrel:1} 
An \emph{evidence relation} is a mapping 
\[
\mathcal{E}: \Tm \to \pow(\eFml)
\]
from terms to sets of epistemic formulas such that for all $s,t\in \textsf{Tm}$, $\alpha \in \eFml$, constants $c_i$, and agents $a_i$:
\begin{multicols}{2}
\begin{enumerate}
\item $\mathcal{E} (s) \cup \mathcal{E}(t) \subseteq \mathcal{E}(s+t)$;
\item $\mathcal{E} (s) \cdot \mathcal{E}(t) \subseteq \mathcal{E}(s \cdot t)$;
\item $t\justifies \mathcal{E}(t) \subseteq \mathcal{E}(!t)$;
\item $c_2 \justifies_{a_2} \cdots  c_n \justifies_{a_n} A \in \mathcal{E}(c_1)$ if $\alpha$ is an axiom;
\item $\alpha \in \mathcal{E}(f^n_t)$, if $\alpha \in \mathcal{E}(f^m_t)$ for $n>m$.
\end{enumerate}
\end{multicols}
\end{definition}

\begin{definition}[Epistemic model]
An  \emph{epistemic model} for\/ $\IPJ$ is a tuple 
$
M = \langle W,R,\mathcal{E},V\rangle
$
where:
\begin{enumerate}
\item $W$ is a non-empty set of objects called worlds.
\item $R$ maps each agent $a$ to a reflexive and transitive accessibility relation $R_a$ on $W$.
\item $\mathcal{E}$ maps each world $w$ and each agent $a$ to an evidence relation~$\mathcal{E}^a_w$.
\item $V$ is a valuation mapping each world to a set of atomic propositions.
\end{enumerate}
\end{definition}

\begin{definition}[Truth within a world] Let $M=  \langle W,R,\mathcal{E},V\rangle$ be an epistemic model for $\IPJ$ and let $w$ be a world in $W$. For an epistemic formula $\alpha \in \eFml,$ we define $M,w \Vdash \alpha$ inductively by:
\begin{enumerate}
\item
$M,w \Vdash \beta \text{ if{f} } \beta \in V(w)$ for $\beta \in \Prop$
\item 
$M,w \Vdash \neg \beta \text{ if{f} } M,w \not \Vdash \beta$
\item 
$M,w \Vdash \beta \land \gamma \text{ if{f} } M,w \Vdash \beta \text{ and } M,w\Vdash \gamma$
\item 
$M,w \Vdash \Box_a \beta \text{ if{f} }  M,u \Vdash \beta \text{ for all }u\in W \text{ with } R_awu$
\item 
$M,w \Vdash t\justifies _a \beta \text{ if{f} }   \beta \in\mathcal{E}^a_w(t) \text{ and } M,u \Vdash \beta \text{ for all }u \in W \text{with } R_awu$
\end{enumerate}
\end{definition}

\begin{definition}[Algebra] Let $U$ be a  non-empty set  and let $H$ be a non-empty subset of $\mathcal{P}(U)$. $H$ will be called an algebra over $U$ if the following hold:
\begin{itemize}
\item $U \in H$
\item $X,Y \in H \rightarrow X \cup Y \in H$ 
\item $X \in H \rightarrow U \setminus X \in H$
\end{itemize}
\end{definition}

\begin{definition}[Finitely additive measure]Let $H$ be an algebra over $U$ and $\mu: H \to S$, where $S$ is the unit interval of the hardy field $\mathbb{Q}[\epsilon]$. We call $\mu$ a \textit{finitely additive measure}  if the following hold:
\begin{enumerate}
\item $\mu(U)=1$
\item 
$
X \cap Y = \emptyset \implies \mu(X \cup Y) = \mu(X)+\mu(Y)
$ for all $X,Y \in H$.
\end{enumerate}
\end{definition}

\begin{definition}[Probability space] A \emph{probability space} is a triple 
$ \langle U, H, \mu \rangle$
where:
\begin{enumerate}
\item $U$  is  a non-empty set
\item $H$ is an algebra over $U$
\item $\mu: H \to S$ is a finitely additive measure 
\end{enumerate}
\end{definition}

\begin{definition}[Quasimodel]
A quasimodel for $\IPJ$ is a tuple 
\begin{equation*}
M = \langle W,R,\mathcal{E},V, U, H, \mu, w_0 \rangle
\end{equation*}
such that
\begin{enumerate}
\item $\langle W,R,\mathcal{E},V\rangle$ is an epistemic model for $\IPJ$
\item $U \subseteq W$
\item $ \langle U, H, \mu \rangle$ is a  probability  space
\item $w_0 \in  U$
\end{enumerate}
\end{definition}

Let  $M = \langle W,R,\mathcal{E},V, U, H, \mu, w_0\rangle$ be a quasimodel, $w\in W$, and $\alpha \in \eFml$.
Since $M$ contains  an epistemic  model, we write $M,w \Vdash \alpha$ for 
\[ 
\langle W,R,\mathcal{E},V\rangle,w  \Vdash \alpha.
\]

\begin{definition}[Events]
Let $M =  \langle W,R,\mathcal{E},V, U, H, \mu, w_0 \rangle$ be a quasimodel.
For an epistemic formula $\alpha \in \eFml$, we define the event that $\alpha$ occurs as
\begin{equation*}
[\alpha]_{M}  := \{u \in U \mid M,u \Vdash \alpha\} 
\end{equation*}
We use $[\alpha]_{M}^C$ for the complement event $U \setminus [\alpha]_{M}$.
\end{definition}
When the quasimodel $M$ is clear from the context, 
we often drop  the subscript~$M$ in $[\alpha]_{M}$.

\begin{definition}[Independent events] Let $M$ be a quasimodel. We say that two events $S,T \in H$ are independent in $M$  if 
\begin{equation*}
\mu(S \cap T) = \mu(S) \cdot \mu(T).
\end{equation*}
\end{definition}

\begin{definition}[Probability almost $r$] 
Let $\langle U,H,\mu \rangle$ be a probability space. For $r\in \mathbb{Q}\cap[0,1]$, we say that  $X \in H$ has probability almost $r$ ($\mu(X)\approx r$) if for all $n \in \mathbb{N}^+$
$
\mu(X) \in \left[ r-\frac{1}{n},r+\frac{1}{n}\right].
$
\end{definition}

\begin{definition}[Truth  in a quasimodel] Let 
\[
M =  \langle W,R,\mathcal{E},V, U, H, \mu, w_0 \rangle
\]
be quasimodel for $\IPJ$. We define $M \models A$ inductively by:
\begin{enumerate}
\item
$M \models A \text{ if{f} } M,w_0 \Vdash A$ for $A \in \eFml$;  otherwise
\item
$M \models \neg B \text{ if{f} } M \not \models B$
\item 
$M \models B \land C \text{ if{f} } M \models B \text{ and } M\models C$
\item
$M \models  \mathcal{P}_{\geq s} \alpha \text{ if{f} } \mu([\alpha]) \geq s$
\item
$M \models \mathcal{P}_{\approx r}\alpha \text{ if{f} } \mu([\alpha]) \approx r $
\end{enumerate}
\end{definition}

\begin{definition}[Measurable model] A quasimodel 
\[
M =  \langle W,R,\mathcal{E},V, U, H, \mu, w_0 \rangle
\] 
is called \emph{measurable} if 
$
[\alpha] \in H  \text{ for all } \alpha \in \eFml
$.
\end{definition}


\begin{definition}[Model]\label{def:model} 
A \emph{model} for $\IPJ$ is a measurable quasimodel $M$ for $\IPJ$ that  satisfies:
\begin{enumerate}
\item $M \models t:_P \alpha \rightarrow \mathcal{P}_{\geq1-\frac{1}{n^k}}(f^n_t:_V\square^P\alpha) \text{ if } n>m \text{ and }\alpha \in \textsf{I}_{m,k};
$
\item $M \models \neg (t:_P \alpha) \rightarrow \mathcal{P}_{\leq\frac{1}{n^k}}(f^n_t:_V\square^P\alpha) \text{ if } n>m \text{ and } \alpha\in \textsf{I}_{m,k}.
$
\end{enumerate}
We say that a formula $A$ is $\IPJ$-valid if $M \models A$  for all models $M$ for $\IPJ$.
\end{definition}

\section{Properties and Results}

We start with two auxiliary lemmas.
\begin{lemma}
Let $\beta,\gamma$ be epistemic formulas.   $\IPJ$ proves 
\begin{enumerate}
\item $\mathcal{P}_{=s} \gamma \to \mathcal{P}_{\leq s}(\gamma \land \beta)$.
\item  $\mathcal{P}_{\leq s}\gamma \land \mathcal{P}_{< r}\beta  \to \mathcal{P}_{< r+s}(\gamma \lor \beta )$ where $r+s\leq 1$.
\end{enumerate}
\end{lemma}
\begin{proof}
For the first claim, suppose  $\mathcal{P}_{=s} \gamma$. Thus we get  $\mathcal{P}_{=1-s} \lnot \gamma$.
Further let~$t$ be such that $\mathcal{P}_{=t} (\lnot \beta \land \gamma)$.
Using axiom~\textsf{(p6)} we infer
\[
\mathcal{P}_{=(1-s)+t}(\lnot \gamma \lor (\lnot \beta \land \gamma)).
\]
Since $(1-s)+t = 1-(s-t)$,  this is equivalent to
\[
\mathcal{P}_{=s-t}(\gamma \land \lnot (\lnot \beta \land \gamma)).
\]
By axiom~\textsf{(p4)} we find 
\[
\mathcal{P}_{=s-t}(\gamma \land \beta).
\]
We conclude $\mathcal{P}_{\leq s}(\gamma \land \beta)$.

To show the second claim, suppose $\mathcal{P}_{\leq s}\gamma$.  By the first claim we get   
\[
\mathcal{P}_{\leq s}(\gamma \land \lnot\beta).
\]
From $\mathcal{P}_{< r}\beta$ 
we obtain using axiom~\textsf{(p6)} that
$\mathcal{P}_{<r+s}((\gamma \land \lnot\beta)\lor \beta)$.  Using axiom~\textsf{(p4)} we conclude  $\mathcal{P}_{<r+s}(\gamma \lor \beta)$.
\end{proof}

We can read the  operator $\mathcal{P}_{\approx 1}$  as \emph{it is almost certain that}.
This operator provably behaves like a normal modality.
\begin{lemma}
Let $\alpha, \beta$ be epistemic formulas.
\begin{enumerate}
\item $\IPJ$ proves $\mathcal{P}_{\approx 1}(\alpha \to \beta) \to ( \mathcal{P}_{\approx 1}\alpha \to  \mathcal{P}_{\approx 1} \beta)$.
\item The rule $\hilbertrule{\alpha}{\mathcal{P}_{\approx 1} \alpha}$ is derivable in  $\IPJ$.
\end{enumerate}
\end{lemma}
\begin{proof}
We first establish that $\IPJ$ proves
\begin{equation}\label{eq:2nd:1}
\mathcal{P}_{\approx 1}(\gamma \lor \beta) \land \mathcal{P}_{\approx 0}\gamma \to \mathcal{P}_{\approx 1}\beta.
\end{equation}
From $\mathcal{P}_{\approx 1}(\gamma \lor \beta)$ we get 
\begin{equation}\label{eq:2nd:2}
\forall r <1 \text{ we have } \mathcal{P}_{\geq r}(\gamma \lor \beta). 
\end{equation}
From $ \mathcal{P}_{\approx 0}\gamma$
we get
\begin{equation}\label{eq:2nd:3}
\forall s > 0 \text{ we have } \mathcal{P}_{\leq s}\gamma.
\end{equation}
From \eqref{eq:2nd:2} and  \eqref{eq:2nd:3} we obtain $ \mathcal{P}_{\approx 1}\beta$.
Suppose towards a contradiction that there exists $r<1$ with $\lnot \mathcal{P}_{\geq r}\beta$.
By the definition of $\mathcal{P}_{<r}$ this is $\mathcal{P}_{< r}\beta$.
Together with \eqref{eq:2nd:3} this yields by the second claim of the previous lemma  that
\[
\mathcal{P}_{< r+s}(\gamma \lor \beta) \quad\forall s > 0 \text{ with $r+s<1$.}
\]
For $s' = \frac{1-r}{2}$ we have $r+s' = \frac{1+r}{2}<1$. Thus there exists $q<1$ with $\mathcal{P}_{< q}(\gamma \lor \beta)$, which contradicts \eqref{eq:2nd:2}.  Hence \eqref{eq:2nd:1} is established.  Let $\gamma$ be $\lnot \alpha$ and observe that 
$\mathcal{P}_{\approx 1}\alpha \to \mathcal{P}_{\approx 0} \lnot \alpha$ is provable in  $\IPJ$. Now the first claim of this lemma immediately follows from  \eqref{eq:2nd:1}.

It remains to show that the rule of $\mathcal{P}_{\approx 1}$ necessitation is derivable. Suppose that~$\alpha$ is derivable.
Thus $\mathcal{P}_{\geq 1} \alpha$ is derivable. Using axioms~\textsf{(p2)} and~\textsf{(p3)}  we obtain $\mathcal{P}_{\geq 1-\frac{1}{n}} \alpha$ for all integers $n$. Thus we infer $\mathcal{P}_{\approx 1}\alpha$.
\end{proof}

An immediate consequence of these lemmas is the following.
If $t$ justifies the prover's knowledge of $\alpha$, then, with almost certainty, the interactive proof protocol based on $t$ will be successful in providing the verifier with a justification for $\alpha$.

\begin{corollary} For $\alpha \in \I$,
$\IPJ$ proves $t\justifies_P \alpha  \rightarrow {P}_{\approx 1}(c\cdot f^\omega_t\justifies_V \alpha)$ for a arbitrary constant~$c$.
\end{corollary}

The deductive system $\IPJ$ is sound with respect to $\IPJ$-models.

\begin{theorem}[Soundness]
Let $\I$ be an arbitrary interaction specification.
For any formula~$F$ we have that
\[
\vdash F \quad\text{implies}\quad \text{$F$ is $\IPJ$-valid.}
\]
\end{theorem}
\begin{proof}
As usual by induction on the length of the derivation. 
The interesting case is when $F$ is an instance of $\mathsf{(c\omega)}$.
But first note that axioms $\mathsf{(m)}$ and  $\mathsf{(c)}$ are $\IPJ$-valid because of Definition~\ref{def:evrel:1} and Definition~\ref{def:model},  respectively.

Now let  $F$ be an instance of $\mathsf{(c\omega)}$.
Then $F$ is of the form
\[
t\justifies_P \alpha  \rightarrow {P}_{\approx 1}(f^\omega_t\justifies_V \Box_P \alpha)
\]
for some $\alpha \in \textsf{I}$.
Let 
$
M =  \langle W,R,\mathcal{E},V, U, H, \mu, w_0 \rangle
$ 
be an arbitrary model for $\IPJ$ and assume $M \models t\justifies_P \alpha$. We need to show
\begin{equation}\label{eq:sound:1}
\mu([f^\omega_t\justifies_V \Box_P \alpha]) \in \left[ 1- \frac{1}{n},1\right] \quad \text{for all $n \in \mathbb{N}^+$}.
\end{equation}
We fix an arbitrary $n \in \mathbb{N}^+$.
Because of  $\alpha \in \I$, we know that there exists an $m$ such that $\alpha \in \I_{m,1}$.
By soundness of axiom $\mathsf{(c)}$
we find that for each $n'>m$
\begin{equation*}
\mu([f^{n'}_t\justifies_V \Box_P \alpha]) \geq 1- \frac{1}{n'}.
\end{equation*}
Let $n''\in  \mathbb{N}$ be such that $n''>m$ and $n''\geq n$. We find
\begin{equation}\label{eq:sound:2}
\mu([f^{n''}_t\justifies_V \Box_P \alpha]) \geq 1- \frac{1}{n''} \geq 1- \frac{1}{n}.
\end{equation}
By soundness of axiom $\mathsf{(m)}$ we get that for each $w \in W$ 
\[
M,w \Vdash f^{n''}_t\justifies_V \Box_P \alpha \text{ implies } M,w \Vdash f^\omega_t\justifies_V \Box_P \alpha.
\]
Therefore, and by finite additivity of $\mu$, we obtain
\begin{equation}\label{eq:sound:3}
\mu([f^\omega_t\justifies_V \Box_P \alpha]) \geq \mu([f^{n''}_t\justifies_V \Box_P \alpha]).
\end{equation}
Taking \eqref{eq:sound:2} and \eqref{eq:sound:3} together yields \eqref{eq:sound:1}.
\end{proof}

In practice, one often considers  interactive proofs systems that are round-based, see~\cite{CompComp}.  
\begin{definition}[Round-based interactive proof system] An interactive protocol $\langle P,V \rangle$ is called \emph{round-based} if the following two conditions hold:
\begin{enumerate}
\item \textsf{Completeness:} Let $x \in \mathcal{L}$. There exists a polynomial $p(x)$  such that the probability that $\langle P,V \rangle$ halts in an accepting state after $p(x)$ many messages is at least $\frac{2}{3}$.
\item \textsf{Soundness:}  Let $x \notin \mathcal{L}$ and let $p(x)$ be any polynomial.  For any interactive Turing machine $P'$, the probability that $\langle P',V \rangle$ halts in an accepting state after $p(x)$ many messages is at most $\frac{1}{3}$.
\end{enumerate}
\end{definition}
This definition achieves negligible (resp.~overwhelming) probabilities by repeating the protocol several times and deciding based on a majority vote.  Although this definition is simple to model in $\IPJ$,  it is not suitable for a limit analysis because our measure is not $\sigma$-additive. 
Note that to properly formalize $\sigma$-additivity one needs countable conjunctions and disjunctions~\cite{IkodinovicOPR20}, which we do not want to include here.
However,  for finitely many rounds,  we can describe how the probability increases throughout the rounds (given that they are pairwise independent). 

\begin{lemma}\label{ipp}
Let $M$ be an $\IPJ$-model for an arbitrary interaction specification~$\I$.
Consider justification terms $s_1,\ldots,s_n$ and an epistemic formula $\alpha$ such that 
\begin{enumerate}
\item $M \models s_i\justifies_V \alpha$ for each $s_i$;
\item $[s_i\justifies_V \alpha]$ and $[s_j\justifies_V \alpha]$ are independent events for all $i \neq j$.
\end{enumerate}
We find that
$
M \models \bigwedge_{i=1,\ldots,n}\mathcal{P}_{\geq1-r}(s_i\justifies_V \alpha) \rightarrow \mathcal{P}_{\geq1 -r^{n} }\alpha.
$
\end{lemma}
\begin{proof}
Whenever $s_i\justifies_V \alpha$ is true at a world $w$,  $\alpha$ is true at $w$ by soundness of axiom $\mathsf{(jt)}$. 
Hence,  by monotonicity of $\mu$ we find

\begin{equation*}
\mu([\alpha]) \geq \mu\left( \bigcup^{n}_{i=1}[s_i\justifies_V \alpha] \right)
=  1 - \mu \left( \bigcap^{n}_{i=1} [s_i\justifies_V \alpha]^C\right)
\overset{indep.}{\geq} 1 - \prod^{n}_{i = 1}r
= 1 - r^{n}
\end{equation*}
\end{proof}

An interactive proof protocol for a language $\mathcal{L}$ has the zero-knowledge property if,  from a successful execution, the verifier only learns that $x$ belongs to $\mathcal{L}$ but nothing else. 
Formally,  a protocol is perfectly zero-knowledge if there exists a probabilistic Turing machine $T$ that generates proof transcripts\footnote{In the setting of interactive Turing machines, a proof transcript is everything that $V$ sees on the public tapes during the protocol.} that are indistinguishable from original ones. If the verifier can obtain additional information with negligible probability, then the protocol is said to be statistically zero-knowledge. 

However, we cannot directly implement this definition because it would require to model the Turing machine $T$ as an agent and we would need to reason about something like indistinguishable terms. 
Simplified, a protocol is zero-knowledge if the verifier cannot compute the prover's secret. In our setting the prover's secret is represented by the term $t$. Hence, $f^n_t\justifies_V t\justifies_P\alpha$ means that the prover's secret has been revealed to the verifier. In fact,  $f^n_t\justifies_V t\justifies_P\alpha$ being unlikely is a direct consequence of the protocol being statistically zero-knowledge because the probability of the verifier knowing the prover's secret is bound by its ability to distinguish between proof transcripts. This gives rise to the following definition of zero-knowledge in $\IPJ$.
\begin{definition}[Evidentially zero-knowledge]
A protocol is \emph{evidentially zero-knowledge} if for all inputs $x$ belonging to~$\mathcal{L}$, the probability of the verifier knowing the prover's evidence for $x $ belonging to $ \mathcal{L}$ is negligible.
\end{definition}
To address evidentially zero-knowledge protocols, we add the following two axioms to $\IPJ$:
\begin{enumerate}
 \item $t\justifies_P \alpha \rightarrow \mathcal{P}_{\leq \frac{1}{n^k}}(f^n_t\justifies_V t\justifies_P\alpha)$ if $n>m$ and $\alpha \in \textsf{I}_{m,k}$;

\item $t\justifies_P \alpha \rightarrow \mathcal{P}_{\approx 0}(f^\omega_t\justifies_V t\justifies_P A) $ if  $\alpha \in \textsf{I}$.
\end{enumerate}
Models for $\IPJ$ are adjusted by requiring the condition:
\begin{equation*}
M \models t\justifies_P \alpha \rightarrow \mathcal{P}_{\leq \frac{1}{n^k}}(f^n_t\justifies_V t\justifies_P\alpha) \text{ if } n>m \text{ and }\alpha \in \textsf{I}_{m,k}.
\end{equation*}
It is easy to show that this extension is sound with respect to its models. The proof of soundness for the second axiom is similar to the soundness proof of~$\mathsf{(c\omega)}$.

\section{Conclusion}

We presented the probabilistic two-agent justification logic $\IPJ$,  in which we can reason about agents that perform interactive proofs.  The foundation of this work is based on probabilistic justification logic combined with interacting evidence systems. 
We further proposed a new technique that asserts a countable axiomatization and makes it possible to reason about the growth rate of a probability measure. 
Intuitively, the set $\I= \bigcap_k \bigcup_m \textsf{I}_{m,k}$ can be thought of as the set of all formulas that are known to be interactively provable. For a formula $\alpha \in \I_{m,k}$ and a term $t$ with $t\justifies_P\alpha$,  
\begin{equation*}
\mathcal{P}_{\geq 1 - \frac{1}{n^{k}}}(f^n_t\justifies_V\Box_P\alpha) 
\end{equation*}
holds for all $n>m$.
Hence, if $\alpha \in \I$,  then the following first order sentence is true
\begin{equation*}
\forall k \exists m \forall (n>m) \mu([f^n_t\justifies_V\Box_P\alpha])\geq 1 - \frac{1}{n^{k}},
\end{equation*}
which is the definition of an overwhelming function. 

Our approach of modelling limits with the help of specification sets is quite versatile as the following example shows.

\begin{example}
Consider a sequence of the form:
\[
 \mathcal{P}_{= L+0.5}(f^1_t \justifies_V \alpha)\quad
 \mathcal{P}_{=L+0.25}(f^2_t\justifies_V \alpha)\quad
 \mathcal{P}_{=L+0.125}(f^3_t\justifies_V \alpha)\quad
\cdots
\]
The sentence  we want to model is:
\begin{equation*}
(\forall \epsilon > 0)(\exists m  \geq 0)(\forall n>m) (\mathcal{P}_{\leq L+\epsilon}(f^n_t\justifies_V \alpha) \land \mathcal{P}_{\geq L-\epsilon}(f^n_t\justifies_V \alpha))
\end{equation*}
Again, for $\epsilon,L \in \mathbb{Q}$ and $m \in \mathbb{N}$, we define sets $\textsf{Conv}^L_{\epsilon,m}$ and let
\[
\textsf{Conv}^L := \bigcap_{\epsilon \in \mathbb{Q}}\bigcup_{m \in \mathbb{N}}\textsf{Conv}^L_{\epsilon,m}.
\]
With the following formulas, we can express that a sequence of probabilities converges:
\begin{enumerate}
\item $\mathcal{P}_{\leq L+\epsilon}(f^n_t\justifies_V \alpha) \land \mathcal{P}_{\geq L-\epsilon}(f^n_t\justifies_V \alpha)$ if $n>m$ and $\alpha \in \textsf{Conv}^L_{\epsilon,m}$;
\item $\mathcal{P}_{\approx  L}(f^\omega_t\justifies_V \alpha)$ if $\alpha \in \textsf{Conv}^L$.
\end{enumerate}
\end{example}
Additionally, we showed that our model can address a round-based definition of interactive proofs, however only for finitely many rounds since our measure is not $\sigma$-additive. 
Further, we also investigated zero-knowledge proofs. As it turns out,  $\IPJ$ cannot model the original definition because we cannot compare justification terms in $\IPJ$. However, we introduced the notion of evidentially zero knowledge,  which fits nicely in our framework.

Moreover, we established soundness of $\IPJ$.  
Our axiomatization is a combination of  systems that are known to be complete and we  conjecture that $\IPJ$ is complete, too.

From a more general perspective, this paper complements the list of motivations for justification logic.
There are the "classical" applications of justification logic in epistemology and proof theory~\cite{Art08RSL,artemovFittingBook,jlbook}.
Recently, justification logic also turned out to be useful to analyze certain deontic situations~\cite{ConflictingObligations} as well as a paradox in quantum physics~\cite{StuderQuantum},  both having to do with certain forms of consistency requirements.
The presented logical analysis of zero knowledge proofs is a novel example that shows the importance of the distinction between explicit (where the justification is shown) and implicit (where the justification is hidden) knowledge. The essence of a zero knowledge proof of a proposition $\alpha$ is that the verifier knows that the prover knows $\alpha$, but the verifier does not know the prover's justification for $\alpha$. 
Thus the verifier does not know why the prover knows $\alpha$ (this hints at  possible connections with the logic of knowing why~\cite{KnowingWhy}). 
That is, the verifier has explicit knowledge of the implicit knowledge of the prover.

\end{document}